\documentclass[aps,prb,superscriptaddress,amsfonts,amsmath,amssymb,floatfix,reprint,longbibliography]{revtex4-2}

\usepackage[utf8]{inputenc}
\usepackage{times}
\usepackage{microtype}
\usepackage{graphicx}
\usepackage{upgreek}
\usepackage{bm}
\usepackage[colorlinks=true, urlcolor=teal, linkcolor=teal, citecolor=teal, pdftex]{hyperref}
\usepackage{xcolor}
\usepackage[normalem]{ulem}
\usepackage{placeins}
\usepackage{amsmath}
\usepackage{amsthm}
\usepackage{dcolumn}
\usepackage{comment}

\usepackage{filecontents}
\usepackage[T1]{fontenc}
\usepackage{csquotes}
\usepackage[german,ngerman,english]{babel}
\usepackage{centernot}
\usepackage{mathtools}
\usepackage{ifthen}
\usepackage{tikz}
\usetikzlibrary{calc}
\usetikzlibrary{decorations.pathreplacing}
\usetikzlibrary{decorations}
\usetikzlibrary{positioning}
\usepackage{pgfplots}
\usepackage{xxcolor}
\definecolor{structure}{rgb}{0.23,0.4,0.7}
\usepackage[normalem]{ulem}

\newtheorem{proposition}{Proposition}
\newsavebox{\blocksavebox}

\definecolor{niceblue}{rgb}{0.33,0.5,0.8}

\newcommand{\cc}{\mathbb{C}}
\newcommand{\ii}{\mathbb{I}}

\newcommand{\refsub}[2]{\hyperref[#1]{\ref*{#1}#2}}

\renewcommand{\min}{\mathchoice{\operatorname*{min}}{\operatorname*{min}}{\mathrm{min}}{\mathrm{min}}}

\newcommand{\norm}[2][]{
  \ifthenelse{\equal{#1}{}}
    {\left\| {#2} \right\|}
    {\ifthenelse{\equal{#1}{uinv}}
      {\left\vert\kern-0.25ex\left\vert\kern-0.25ex\left\vert {#2} \right\vert\kern-0.25ex\right\vert\kern-0.25ex\right\vert}
      {\left\| {#2} \right\|_{#1}}
    }
}

\newcommand{\taverage}[2][]{
  \ifthenelse{\equal{#1}{}}
  {\overline{#2}}
  {\overline{#2}^{#1}}
}

\newcommand{\tracedistance}[3][]{
  \ifthenelse{\equal{#2}{}}
  {\ifthenelse{\equal{#3}{}}
    {\mathcal{D}_{#1}}{}
  }{
    \ifthenelse{\equal{#1}{}}
    {\mathchoice{\operatorname{\mathcal{D}}\left(#2,#3\right)}{\operatorname{\mathcal{D}}(#2,#3)}{\operatorname{\mathcal{D}}(#2,#3)}{\operatorname{\mathcal{D}}(#2,#3)}}
    {\mathchoice{\operatorname{\mathcal{D}}_{#1}\left(#2,#3\right)}{\operatorname{\mathcal{D}}_{#1}(#2,#3)}{\operatorname{\mathcal{D}}_{#1}(#2,#3)}{\operatorname{\mathcal{D}}_{#1}(#2,#3)}}
  }
}
\newcommand{\fidelity}[3][]{
  \ifthenelse{\equal{#2}{}}
  {\ifthenelse{\equal{#3}{}}
    {\mathcal{F}_{#1}}{}
  }{
    \ifthenelse{\equal{#1}{}}
    {\mathchoice{\operatorname{\mathcal{F}}\left(#2,#3\right)}{\operatorname{\mathcal{F}}(#2,#3)}{\operatorname{\mathcal{F}}(#2,#3)}{\operatorname{\mathcal{F}}(#2,#3)}}
    {\mathchoice{\operatorname{\mathcal{F}}_{#1}\left(#2,#3\right)}{\operatorname{\mathcal{F}}_{#1}(#2,#3)}{\operatorname{\mathcal{F}}_{#1}(#2,#3)}{\operatorname{\mathcal{F}}_{#1}(#2,#3)}}
  }
}

\newcommand{\Sr}[3][]{
  \ifthenelse{\equal{#1}{}}
    {\operatorname{\mathnormal{S}}(#2\|#3)}
    {\operatorname{\mathnormal{S}}_{#1}(#2\|#3)}
}

\DeclareMathOperator{\1}{\mathbb{I}}



\newcommand{\tr}{\text{tr}}

\definecolor{jens}{rgb}{0.1,0.5,0.1}
\definecolor{martin}{rgb}{0,0,1.0}

\newcommand{\new}[1]{{\color{black} #1}}

\newcommand{\beq}[0]{\begin{equation}}
\newcommand{\eeq}[0]{\end{equation}}

\newcommand{\hide}[1]{}

\begin{document}

\title{\new{Lower} bounds to variational problems with guarantees}

\author{J.\ Eisert}
\address{Dahlem Center for Complex Quantum Systems, Freie Universit{\"a}t Berlin, 14195 Berlin, Germany}
\address{Helmholtz-Zentrum Berlin f{\"u}r Materialien und Energie, 14109 Berlin, Germany}

\begin{abstract}
Variational methods play an important role in the study of quantum \new{many-body problems}, both in the \new{flavor} of classical variational principles based on tensor networks as well as of quantum variational principles in near-term quantum computing. This \new{work} stresses that for translationally invariant lattice Hamiltonians with periodic boundary conditions, one can easily derive efficiently computable lower bounds to ground state energies that can and should be compared with variational principles providing upper bounds. As small technical results, it is shown that (i) the Anderson bound and a (ii) common hierarchy of semi-definite relaxations both provide approximations with performance guarantees that scale like a constant in the energy density for cubic lattices. (iii) Also, the Anderson bound is systematically improved as a hierarchy of semi-definite relaxations inspired by the \new{quantum} marginal problem. 
\end{abstract}

\maketitle

\section{Introduction}

Variational principles are for good reason ubiquitous in the study of quantum \new{many-body} system. 
They allow for achieving insights into the physics of quantum many-body systems in ways that are 
hard to obtain by any other method, in particular in situations when strong correlations are dominant 
or when an instance of the sign-problem occurs so that other work-horses of the numerical study of quantum many-body systems such as density functional theory \cite{DFT3} or quantum Monte
Carlo methods \cite{QuantumMonteCarloMethods} are being challenged in their performance.
 A ground state of a \emph{local 
many-body Hamiltonian} $H_N$ defined on a lattice ${\cal L}$
composed of $N$ sites or vertices each of which is associated with
a $d$-dimensional quantum system is even defined as a state satisfying
\begin{equation}\label{LH}
\rho_G:= \text{argmin}_{\rho\in {\cal S}((\cc^d)^{\otimes N})} {\tr} (\rho H_N).
\end{equation}
This is the solution of a variational principle over all quantum states $\rho$ defined on these
$N$ degrees of freedom. A substantial body of literature on the study
of strongly correlated quantum systems with many degrees of freedom is hence dedicated
to formulating meaningful tractable ansatz classes to tackle this general variational principle.
Naturally, if one merely optimizes over certain families of 
quantum states $\rho\in {\cal T}((\cc^d)^{\otimes N})$,
the latter being a suitable subset of all quantum states, one arrives at quantum 
states providing an 
\emph{upper bound} to the ground state energy density
\begin{equation}
	e_\text{min}(H_N):= \frac{1}{N}
	\min_{\rho\in {\cal S}((\cc^d)^{\otimes N})}
	\text{tr}(\rho H_N)  = \frac{1}{N}\lambda_\text{min}(H_N).
\end{equation}
This idea has presumably become most prominent
in the study of \emph{tensor network states} 
\cite{Orus-AnnPhys-2014,RevModPhys.93.045003,MPSRev}
where classically efficient (in memory storage, 
but also in computational complexity at least in approximation) 
ansatz classes are being
used to generate excellent approximations of true ground states. The basis of their
functioning is that common ground states of local Hamiltonians are much less entangled
than they could be \cite{AreaReview}, allowing for efficient approximations. Indeed, 
ground states of gapped one-dimensional local Hamiltonians can basically be parametrized 
by the solutions of such variational principles, in fact by instances of matrix product states.

An alternative ansatz that is increasingly
becoming popular is that of using quantum circuits in what is called the \emph{variational
quantum eigensolver} \cite{McClean_2016,bharti_2021_noisy}: In this context, one thinks of near-term 
quantum computers that have the ability to prepare states from a parametrized family
of quantum states, for which the expectation value of the given Hamiltonian is then estimated and computed 
from measurement data. A limitation of such ansatzes is the often relatively low expressivity of the
parametrized family of states; at the same time, in contrast to classical approaches
they do not require an efficient classical contraction. \new{What is more, the training of such variational quantum algorithms based on expectation values requires many samples \cite{Gradients}, gradients can be small so that 
training is difficult
\cite{BarrenPlateaus,ReviewBarrenPlateaus}, and in instances, they can also be ``dequantized'' by proposing classical algorithms for the same task. That said, they remain an important application of near-term quantum computers
\cite{TrainabilityDequantization,DoesBarrenSimulability}.}
Both ansatzes deliver upper bounds to
the energy density of the ground state.

Either way, as such, by construction such methods do not provide any certificate of the quality of the 
approximation of the ground state energy. This \new{may be} less of an issue for tensor
network methods that reach enormous precision, \new{but even there certificates are helpful.}
\new{It surely} applies to quantum variational approaches
that are presently within reach. This short 
%pedagogical note 
\new{manuscript}
stresses that lower bounds of the ground state 
energy that provide precisely such certificates can be found, often with little
 programming effort 
\cite{PhysRev.83.1260,Mazziotti,PhysRevLett.108.200404,BaumgratzLowerBounds,PironioLowerBounds}. 
A similar point has also been made for variational quantum
eigensolvers in Ref.~\cite{PhysRevResearch.2.043163}. 
 For simplicity, 
throughout this \new{work}, we consider
\emph{translationally invariant} local Hamiltonians with nearest-neighbour interactions,
defined on cubic lattices ${\cal L}$ of size $|{\cal L}|:= N=n^D$, 
naturally equipped with periodic boundary conditions. This means that the Hamiltonian 
takes the form
\begin{equation}\label{LH}
H_N = \sum_{j\in {\cal L}} \tau_j (h),
\end{equation}
where $\tau_j$ places the nearest-neighbour Hamiltonian term $h$ at a root 
site $j\in {\cal L}$ of a lattice hosting a $d$-dimensional quantum degree of freedom. \new{For conceptual simplicity, we assume $h$ to be supported on $2^D$ sites.}
 
\new{In this work}, it is shown that common such lower bounds
can be seen to feature performance guarantees that scale as $O(1)$ in the system size $N$ 
for the ground state energy density with a small constant that can be tightly bounded.
Also, the Anderson bound is improved.
The arguments presented here are all elementary, but on a conceptual level, it is still worth stressing 
that any variational eigensolver has to deliver
a value that is more accurate than a constant in $N$ in order to deliver a meaningful estimate for the
ground state energy. In this sense, the results stated here can be seen as (immediate instances of) 
``de-quantization'' results,
in that they place  stringent demands on any quantum algorithm aimed at obtaining a quantum advantage
when estimating ground state energies.
For what follows, we define as the central quantity of this \new{work}
\begin{equation} 
e_\text{min} := \limsup_{n\rightarrow \infty} e_\text{min}(H_{n^D})
\end{equation}
as the asymptotic ground state energy density.
%, so energy per site.

%The Anderson bound gives rise to a lower bound. Interestingly, it gives rise to a good bound.
%Primarily pedagogical note. Small constant error per site. Quantum ... variational principles... must
%be better than this... in $P$. Quantum PCP... not on such lattices.
%The correction is also order $N$, so contributes. Sometimes underappreciated.
%NEED POSITIVE TERMS! Shift. But then, if one shifts to exactly zero ground state, then...
%not so bad. The exact energy is $-0.886294$. And I find $-0.8516$. But this is larger.

%Anderson: On the other hand,
%one can easily see that the least eigenvalue of the total hamiltonian
%must be greater than the sum of the least eigenvalues of its parts:
%Sure, take the sum. What to do with the boundary term? 

%\emph{A performance guarantee for the Anderson bound.} 
\section{A performance guarantee for the Anderson bound}
The Anderson bound  \cite{PhysRev.83.1260}
is a remarkably
simple lower bound to the ground state energy of quantum many-body Hamiltonians,
basically merely exploiting the triangle inequality of the operator norm $\|.\|$ applied to the
Hamiltonian equipped with a negative sign. It is conceptually easy and is 
implementable with a small programming effort, of less than an hour for a one-dimensional 
Hamiltonian problem. The performance of the bound is depicted in Fig.~2 for the 
\emph{Heisenberg Hamiltonian} in one spatial dimension. 
In this section, we see that it actually always delivers an approximation
of the ground state energy density up to a small constant in $N$, in fact, arbitrarily
small, with a computational effort that is exponential in $m$.

\begin{proposition}[Performance guarantee of the Anderson bound] 
Consider a family of translationally invariant Hamiltonians of the 
form (\ref{LH}) on a cubic lattice in some spatial dimension $D$ \new{with periodic boundary conditions}, indexed by the system size $N=n^D$, 
and let $\lambda_\text{min}(h_m)$ be the smallest eigenvalue of a cubic patch $h_m$ of $H_N$ on $m^D$
sites, with open boundary conditions, then
%, for $h\geq 0$,
\begin{eqnarray}\label{Ande}
	A(m,D)&:= &\frac{\lambda_\text{min}(h_m)}{(m-1)^D}\leq e_\text{min}, \\
%	\left | e_\text{min} -
%	\frac{\lambda_\text{min}(h_m)}{(m-1)^D}
%	A(m,D)
%	\right| &\leq &
%	\frac{D}{m}\|h\|\nonumber\\
%	- \lambda_\text{min}(h_m)
%	\biggl[
%	\frac{1}{(m-1)^D}
%	&-&
%	\frac{1}{m^D} 
%	\biggr].\\
	\left | e_\text{min} -
	A(m,D)
	\right| &\leq & 
	\frac{D}{m}\|h\| 
	- \lambda_\text{min}(h_m)
	\biggl[
	\frac{1}{(m-1)^D}
	-
	\frac{1}{m^D} 
	\biggr].\nonumber\\
\end{eqnarray}
\end{proposition}

\begin{proof} The first inequality, first stated in Ref.~\cite{PhysRev.83.1260}
and here adapted to the asymptotic limit of large cubic lattices,
is an immediate consequence of the following basic and still profound insight: One composes a Hamiltonian of
$N = [(m-1) J]^D $ sites into 
\new{$J^D$} overlapping parts (see Fig.~1(a))
\begin{equation}
	H_N = \sum_{s\in I_{m,J}} 
%    \new{\sum_{s\in I_{j_1,\dots, j_D}}}    
\tau_s(h_m)
\end{equation}
where 
%\new{$I_{j_1,\dots, j_D}:=\{((m-1)(j_1-1)+k_1,\dots,(m-1)(j_D-1)+k_D):k_1,\dots, k_D=1,\dots, m-1\}$}.
$I_{m,J}:= \{ ((m-1)( j_1-1)+1,\dots, (m-1) (j_D-1)+1): j_1,\dots, j_D=1,\dots, J\}\subset{\cal L}$. 
Then, clearly,
\begin{equation}
	\lambda_\text{min} (H_N) \geq   
    \sum_{s\in I_{m,J}}  \lambda_\text{min} (h_m)  
    = J^D \lambda_\text{min} (h_m) ,
\end{equation}
where the first bound follows from the fact that the smallest eigenvalue of $H_N$ is lower bounded
by the sum of the smallest eigenvalues of each of the $J^D$ parts consisting of
%$(m-1)^D $ 
$\new{m^D}$
sites each. For $J\rightarrow\infty$. 
this gives the first statement of Eq.~(\ref{Ande}).
The performance guarantee
can be shown by considering a different partition (see Fig.~1(b)),
\begin{equation}\label{parti}
	H_N = \sum_{s\in K_{m,J}} \tau_s(h_m)+ V_N,
\end{equation}
where $K_{m,J}:= \{ (m (j_1-1) +1,\dots,  m (j_D-1)+1): j_1,\dots, j_D=1,\dots, J\}\subset {\cal L}$, 
and hence %$N = m^D J^D$,
\new{$n=Jm$, $N=n^D$,}
and where $V_N$ is the remainder
term that consists of nearest neighbour terms connecting the slightly larger patches. 
Let us define $|\phi \rangle := \text{argmin}_{|  \psi\rangle} {\langle\psi | h_m|  \psi\rangle}$.
Then
\begin{eqnarray}
	&&\frac{\lambda_\text{min}(H_N)}{N} - 
	\frac{\lambda_\text{min}(h_m) }{(m-1)^{D}} =
	\min_{|  \psi\rangle} \frac{\langle\psi|H_N |\psi\rangle}{N} 
	-\frac{\langle \phi | h_m |\phi\rangle} {(m-1)^D} \nonumber
	\\
	& \leq & \langle \phi | h_m |\phi\rangle
	\frac{J^D}{N} +
	 \frac{1}{N} \|V_N\| -  \langle \phi | h_m |\phi\rangle
	 \frac{1}{(m-1)^D}\nonumber
	 \\
	& \leq & \langle \phi | h_m |\phi\rangle
	\frac{J^D}{N} +
	 m^{D-1}\|h\|\frac{D J^D}{N} -  \langle \phi | h_m |\phi\rangle
	 \frac{1}{(m-1)^D}\nonumber
	 \\
	&= &
	\lambda_\text{min} (h_m)
	\biggl[
	\frac{1}{m^D}
	-
	\frac{1}{(m-1)^D} 
	\biggr]+ \frac{D }{m }\|h\|,
%	\min_{|  \psi\rangle} \frac{\langle\psi|H_N |\psi\rangle}{ m^D J^D}
%	-  \frac{\langle \phi | h_m |\phi\rangle} {m^D}
%	+ \frac{\langle \phi | h_m |\phi\rangle} {m^D}
%	-\frac{\langle \phi | h_m |\phi\rangle} {(m-1)^D} \nonumber \\
%	&\leq &
%	\frac{1}{m^D} 
%	 \frac{\| V_N \| }{J^D m^D}
%	+ \left[
%	 \frac{1}{m^D}- \frac{1}{(m-1)^D}\right]
%	\lambda_\text{min}(h_m) \nonumber\\
%	&\leq & D m^{D-1}
%	\frac{\|h\|}{m^D} = O(1)
%	,
\end{eqnarray}
where the first inequality follows from the fact that the minimum of $\min(\rho h_m)$ over quantum
states $\rho$ takes the smallest
value for $\langle \phi| h_m|\phi\rangle$, the second from the 
triangle inequality of the operator norm. Then, one encounters \new{fewer than} $D m^{D-1}$ boundary terms in a cubic patch in $D$ dimensions involving $m^D$ many vertices (\new{avoiding double counting}); and again, 
the
triangle inequality of the operator norm is used.
As before, the statement follows for $J\rightarrow\infty$.
%and the property that ${1}/{m^D}- {1}/{(m-1)^D}<0$ for positive integers $m$. 
%The assumption of $h\geq 0$ is not restricting generality,
%as the spectrum of the individual terms can be arbitrarily shifted. 
\end{proof}%

\begin{figure}[t]
        \includegraphics[width = .75\columnwidth]{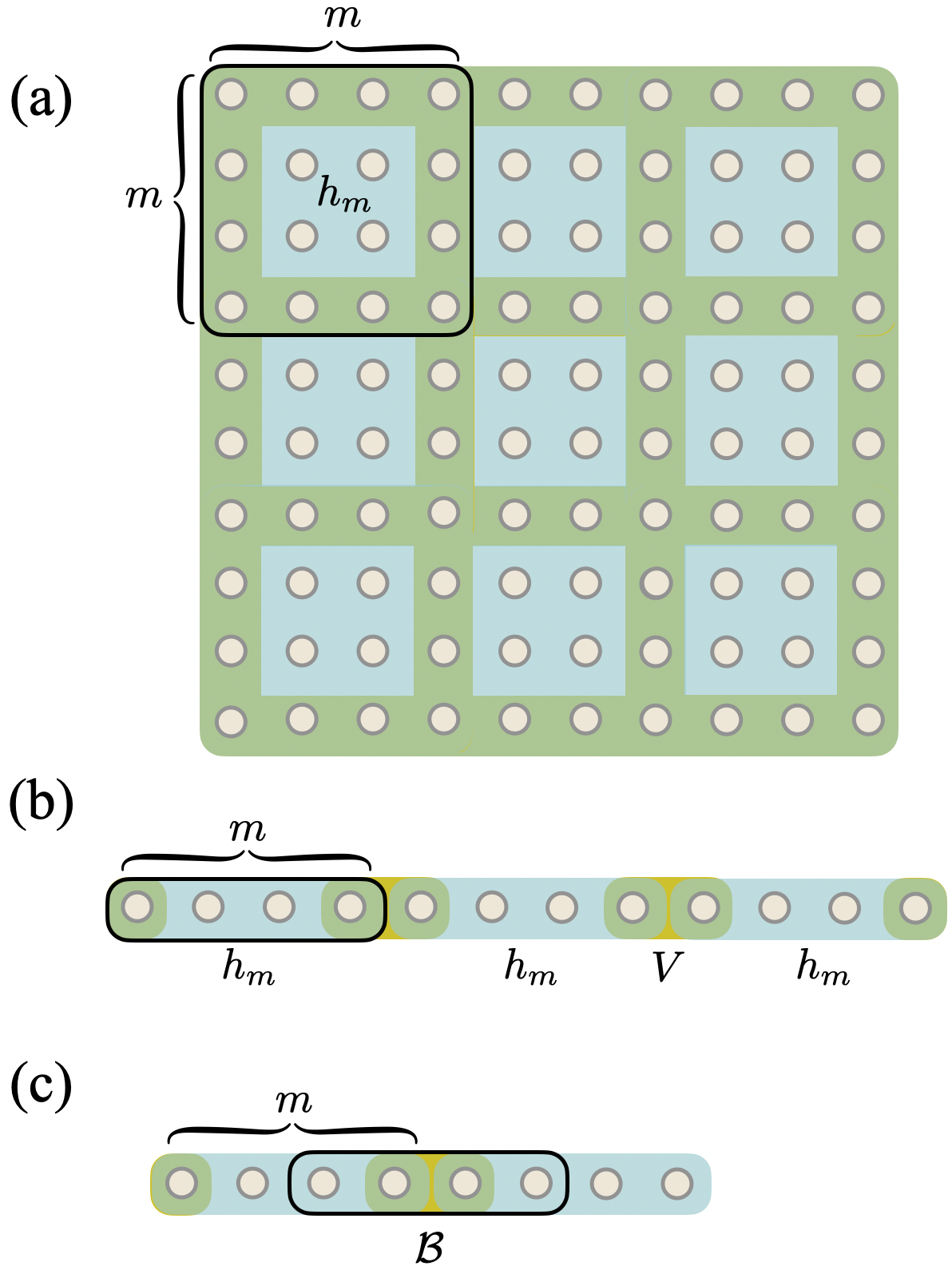}
        \caption{(a) The configuration in the original Anderson bound applied to a two-dimensional
        translationally invariant lattice system. \new{For the periodic boundary conditions, the last row and column are identified with hte first one.} (b) The configuration in one spatial dimension used to show the
        performance guarantee of the Anderson bound. (c) The configuration employed for the improved
        Anderson bound based on semi-definite programming and the marginal problem, 
        again applied to one spatial dimension.}
\end{figure}

%\emph{Performance guarantees for semi-definite relaxations.} 
\section{Performance guarantees for semi-definite relaxations}

Similar 
performance guarantees can be shown for common hierarchies of semi-definite relaxations
\cite{Boyd2004} of finding ground states of local Hamiltonians \cite{Mazziotti,PhysRevLett.108.200404,BaumgratzLowerBounds,PironioLowerBounds},
\new{giving rise to in practice commonly much tighter bounds than those provided by Anderson-type bounds \cite{PhysRevLett.108.200404}.}
The core idea of these approaches is very simple: The constraint 
of quantum states being positive semi-definite
$\rho\geq0$ is relaxed to operators $\omega$ satisfying 
$\text{tr}(\omega O^\dagger O)\geq 0$
for suitable operators $O$. In a next step, the quantum state $\omega$ is eliminated. 
Since the constraint of quantum states being positive semi-definite
is relaxed, one naturally arrives at lower bounds to ground state energies.

\begin{figure}[t]
        \includegraphics[width = 0.78\columnwidth]{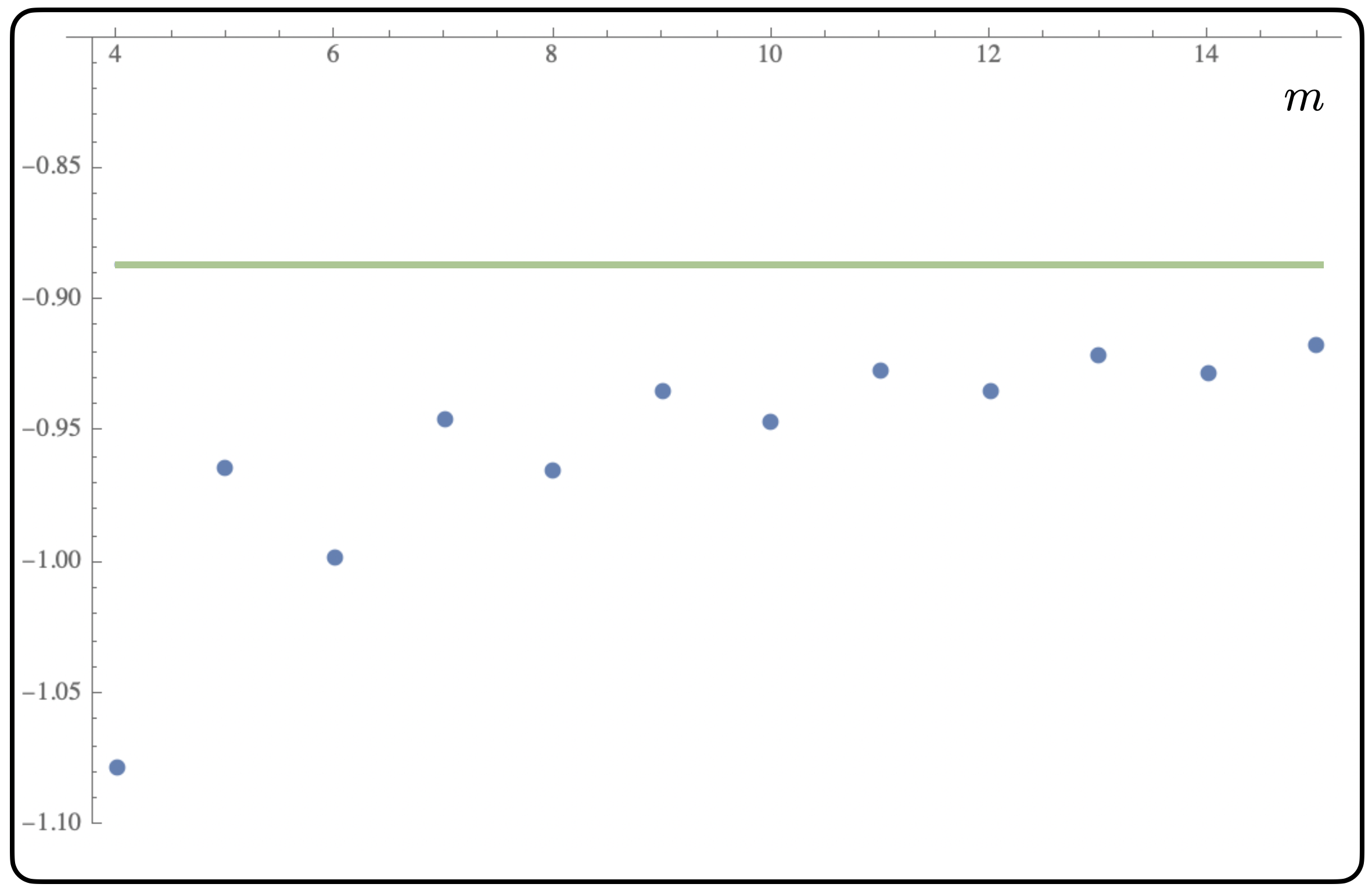}
        \caption{The Anderson bound for the one-dimensional Heisenberg model 
        $H = \frac{1}{2} \sum_{j}\tau_j (X\otimes X+ Y\otimes Y+ Z\otimes Z)$ as a function of the patch size $m$
        until $m=15$,
        featuring a noteworthy even-odd effect. The straight line represents the exact
        ground state energy density, $e_\text{min}=-2\log(2) +1/2$.}
     \end{figure}

Concretely, for a distinguished lattice site $j$, 
consider a set ${\cal M}$ of cardinality $|{\cal M}|=:C$ 
of operators $O_a\in {\cal M}$ 
%defined on $(\mathbb{C}^d)^{\otimes N}$ 
that 
has the property that the Hamiltonian term on this site can be written as 
\begin{equation}
\tau_j(h)= \sum_{a,b} c_{a,b} O_a^\dagger O_b
\end{equation}
with $c_{a,b}\in \cc$ for $a,b=1,\dots, C$. The set of sites on which all operators in ${\cal M}$ are non-trivially
supported 
is denoted as ${\cal S}$, which at the same time contains by construction the support of $\tau_j(h)$
(but which may be substantially larger).
This set is being seen as the root set of operators that then 
acts in a translationally invariant fashion on each lattice site in the same fashion.
The methods discussed in Refs.~\cite{Mazziotti,PhysRevLett.108.200404,BaumgratzLowerBounds,PironioLowerBounds}
essentially amount to identifying suitable such sets of operators.
The operators considered in ${\cal M}$ will feature algebraic relations (such as commutation or
anti-commutation relations).
%
%The operators $O_j\in {\cal O}$ will give rise to
%linear constraints on $X$ arising from commutation or anti-commutation 
%relations, and neglecting some of those will again only give rise to a lower bound of the 
%ground state energy density. 
%For the set ${\cal M}_N$, we denote
%with ${\cal X}_N$ the set of linear constraints, extended to the translationally
%invariant setting to ${\cal Y}_N$.
%
%Keep constraints per site. 
%
%No, take $X$ to be translationally invariant.
We consider again translationally invariant settings,
with 
\begin{equation}
{\cal O}_N:= \{ \tau_j (O_a): \forall j\in {\cal L}, a=1,\dots,  C\} 
\end{equation} 
being defined as  translates
of the root set of operators ${\cal M}$, giving rise to a set of cardinality $|{\cal O}_N|=: D_N$.
To express the lower bound, we define the matrix $X$ with
\begin{equation}
X_{a,b}:= \text{tr}(O_a^\dagger O_b \omega)
\end{equation}
for $a,b=1,\dots ,D_N$.
The matrices $X\in\cc^{D_N\times D_N}$ will again be taken to be translationally invariant, in that 
 \begin{equation}\label{TIC}
X_{a,b}= \text{tr}(\tau_j(O_a^\dagger) \tau_j(O_b )\omega)
 \end{equation}
 for all $j\in {\cal L}$ and all $a,b=1,\dots, D_N$. 
 The above constraints between the operators will be reflected by
expressions of the form
% \begin{equation}
 	$\text{tr} (X R) = 1$
% \end{equation}
for suitable matrices $R\in\cc^{D_N\times D_N}$.
%, giving rise to a set ${\cal R}_N$.
%
 The constraints 
 %$R\in {\cal R}_N$ 
 immediately also act in a 
 translationally invariant fashion: We denote the set of matrices $R\in \cc^{D_N\times D_N}$
 that reflect both the local algebraic constraints 
 %incorporated by ${\cal R}_N$ 
 as well as 
 the translational invariance in Eq.\ (\ref{TIC}) as $\text{tr} (X R) = 1$ by ${\cal X}_N$.
%The set of constraints considered will be denoted as ${\cal R}_N$, which is 
%naturally again reflecting translationally invariant constraints,
%in that for every  $X$ satisfying $\text{tr} (X R) = 1$ for a suitable
%$R\in {\cal R}_N$ and every $X'$ with
% \begin{equation}
%X'_{a,b}:= \text{tr}(\tau_j(O_a^\dagger) \tau_j(O_b )\omega)
% \end{equation}
% for some $j\in n^D$ and $a,b=1,\dots,D_N$,
%there exists an $R'\in {\cal R}_N$ with $\text{tr} (X' R') = 1$.
%Written in these general terms, the guarantee is an immediate
%consequence of the construction. 
%If one minimizes over 
%We denote with ${\cal P}_N({\cal O}_N,{\cal R}_N )$ 
%the convex set of 
%all operators $\omega$ defined on $(\mathbb{C}^d)^{\otimes N}$
%that satisfy $\text{tr}(\omega)=1$ and
For all $X\in \cc^{D_N\times D_N}$ that satisfy $X\geq 0$ and $\text{tr} (X R) = 1$ for all  $R\in{\cal X}_N$,
\begin{equation}\label{B}
	\sum_{a,b}\alpha_a^\ast \alpha_b
	\text{tr}(O_a^\dagger O_b \rho)
	=\sum_{a,b} \alpha_a^\ast \alpha_b
	X_{a,b}
	\geq 0
\end{equation}  
will hold true for all $\alpha\in {\mathbb{C}}^{D_N}$, allowing to 
devise a lower bound to the 
ground state energy density.  In fact, these bounds will again
be lower bounds with a guaranteed constant error.
%The lower bound
%\begin{equation}
%	\frac{1}{N}
%	\inf_{\rho\in {\cal S}_N}
%	\text{tr}(\rho H_N) \geq \frac{1}{N}
%	\inf_{\omega \in {\cal P}_N({\cal O}_N,{\cal R}_N)}
%	\text{tr}(\omega H_N)
%\end{equation} 
%follows by construction, as the set of positive semi-definite states is relaxed to the 
%convex set ${\cal P}_N({\cal O}_N,{\cal R}_N)\supset {\cal S}((\mathbb{C}^d)^{\otimes N})$.

\begin{proposition}[Performance guarantee of semi-definite bounds] 
The solution $x_N$ of the semi-definite problem
\begin{eqnarray}
	\text{minimize} & \text{tr} (h  X),\label{l1}\\
{\text{subject to }} &X\geq 0,\label{l2}\\
&\text{tr}(X R) =1 \, \forall R\in {\cal X}_N,\label{l3}
\end{eqnarray}  
%\begin{equation}
%e_\text{min}(H_N) \leq \sum_{a,b}...
%\end{equation}
where $ {\cal X}_N$ reflects algebraic constraints as well as translational invariance,
satisfies $x_N\leq e_\text{min}(H_N)\leq x_N+ O(1)$. 
\end{proposition}

\begin{proof} Since (\ref{B}) is true for all $\alpha\in \cc^{D_N}$ exactly if
$X\geq 0$, and since
\begin{equation}
%	\inf_{\rho\in {\cal S}_N}
%	\text{tr}(\rho H_N) =
	%\frac{1}{N}
%	\inf
\sum_{a,b} h_{a,b} 
	\text{tr} (O_a O_b \omega) = 
	\sum_{a,b} h_{a,b} 
	X_{a,b} = 
%	\frac{1}{N}
	\text{tr} (h  X),
\end{equation}  
we get the above lower bounds as the solution of
the convex optimization problem Eqs.\ (\ref{l1})-(\ref{l3}), as the energy 
minimization problem is relaxed
to a semi-definite problem. 
We will continue the argument by showing that these
bounds will scale like a constant in the energy density. 
For this, we consider a set of fixed lattice sites ${\cal T}$ independent of $N$ that is a superset
of the set ${\cal S}$ that hosts $\tau_j(h)$ but which may be substantially larger than its support,
 ${\cal S}\subset {\cal T}$.
The strategy will be to 
construct bounds that are even lower bounding the ones from Proposition 2, but that already give rise to
a constant energy approximation from below. 
These lower bounds are given by the solution of semi-definite problems with
(\ref{l1}) and (\ref{l3}), where (\ref{l2}) is relaxed to 
the principal sub-matrix associated with ${\cal T}$ satisfying
\begin{equation}
	\left.X\right|_{\cal T}\geq 0.
\end{equation}
This will give rise to lower bounds of the original semi-definite optimization problem, 
as $X\geq 0$ implies that the principal sub-matrix
$\left.X\right|_{\cal T}$ 
\new{of the matrix $X$} is also positive semi-definite. Since no constraint in the problem involves the system size
any longer,
and the only constraints dependent on $N$ in  ${\cal X}_N$ enforce translational invariance,
the solution of this new semi-definite problem 
\begin{eqnarray}
	\text{minimize} & \text{tr} (h  X),\label{c1}\\
{\text{subject to }} &\left.X\right|_{\cal T}\geq 0,\label{c2}\\
&\text{tr}(X R) =1 \, \forall R\in {\cal X}_N\label{c3}
\end{eqnarray}  
scales like $O(1)$
in $N$. This  implies that the solution of the original semi-definite problem in Proposition 2
satisfies
$x_N \leq  e_{\rm min}(H_N ) \leq x_N + O(1)$. At the same time, it is clear that by enlarging
the set ${\cal M}$, the actual ground state energy $e_{\rm min}(H_N ) $ 
can be arbitrarily well
approximated from below.
\end{proof}

%\emph{A hierarchy of improved Anderson bounds.} 
\section{A hierarchy of improved Anderson bounds}
The Anderson bound as is provides strikingly
good lower bounds of the energy density up to small constant errors at very little effort. For this
reason, the question arises whether it can be systematically improved. 
Resorting to the \emph{quantum marginal problem}, one can\new{,} in fact\new{,} %easily
\new{improve} the Anderson bound. For this, consider again %the
\new{a} configuration \new{similar to the one} used in Proposition 1, described in Eq.~(\ref{parti}), \new{and in fact employ an 
Anderson bound for twice the patch size. This amounts to the} 
%%
%
%Obviously, the full optimization problem minimizing the 
%ground state energy density $e_\text{min}(H_N)$ can be written as %the 
solution to 
the convex
optimization problem
\begin{eqnarray}
\label{fu3}
	\text{minimize} & \text{tr} (\omega h_{2m})\label{fu1}\\
	{\text{subject to }}&
	%\omega = \sigma,\label{fu2}\\
\omega\geq 0,\, \text{tr}(\omega)=1
%&\tau_j(\omega) = \tau_k(\omega)\, \forall j,k\in {\cal L},
\label{fu4}
\end{eqnarray}
%\begin{eqnarray}
%	\text{minimize} & \sum_{s\in K_{m,J}} \text{tr} (\omega %\tau_s(h_m)) + \text{tr}(\sigma V_N),\label{fu1}\\%
%	{\text{subject to }}&
%	\omega = \sigma,\label{fu2}\\
%&\omega,\sigma\geq 0,\, \text{tr}(\omega)=\text{tr}(\sigma)=1,%\label{fu3}\\
%&\tau_j(\omega) = \tau_k(\omega)\, \forall j,k\in {\cal L},
%\label{fu4}
%\end{eqnarray}
over \new{quantum} states defined on \new{$(\mathbb{C}^d)^{\otimes {2m}}$. For any such bound, the computational effort is exponential in $m$, so while the ground state of $h_m$ may be within reach, that of 
$h_{2m}$ may not be. 
Inspired by the \emph{quantum marginal problem} \cite{QuantumMarginalProblem,ChristandlMarginal}, this problem can, however, be relaxed in the following way, to a} family of efficiently solvable semi-definite problems (see also Fig.~1(c), \new{where the overlapping sites are referred to as ${\cal B}$}) that strictly 
generalize the
Anderson bound.  
%an optimization problem that can, needless to say, be only solved with exponential computational 
%effort. 
%This can be relaxed,
%however, to a family of efficiently solvable semi-definite problems (see also Fig.~1(c)) that strictly 
%generalize the
%Anderson bound.
%\new{In what follows, we consider two patches of $h_m$ connected by a coupling term $h$. This }

\begin{proposition}[Improved Anderson bounds] \new{For a one-dimensional translationally invariant Hamiltonian $H_N$ and
an integer $s$ with $s\leq m$, the convex relaxation
of the optimization problem in Eqs.\
(\ref{fu3},\ref{fu4}) can for 
$s=1,\dots, m$ 
be relaxed to the semi-definite optimization 
problem
\begin{eqnarray}
\text{minimize} & 2 \text{tr}(\omega h_m) +
 \text{tr}((\ii\otimes h \otimes \ii)\sigma),\\
 & \sigma|_{\{1,\dots, s\}} 
 = \omega|_{\{
 m-s+1,\dots, m
 \} },\\
 & \sigma|_{\{2+1,\dots, 2s\} } = \omega|_{\{
 1,\dots, s
 \} },\\
& \sigma\geq 0, \omega\geq 0,
 \text{tr}(\sigma)= \text{tr}(\omega)=1,
\end{eqnarray}
the optimal value
$x_{m,s}$ of which
satisfies
\begin{equation}
\frac{x_{m,s}}{(2m-1)} 
\leq \frac{x_{m,m}}{(2m-1)} \leq e_\text{min}.
\end{equation}
}
\end{proposition}

%\begin{proposition}[Improved Anderson bounds] Consider for a one-dimensional translationally invariant Hamiltonian $H_N$  and an integer $m$ the term $h_m$ defined on $m$ sites as before. 

%For an integer $s$ with $2s\leq m$ consider
%the set of sites ${\cal B}:= \{m+j \mod m: j=-s+1, -s+2,\dots, s\}$. Then
%the solution $z_{m,s}$ of the semi-definite optimization problem
%\begin{eqnarray}
%	\text{minimize} &  \text{tr} (\omega h_m) + \text{tr}(\sigma (\mathbb{I}\otimes h)),\\
%	{\text{subject to }}&
%	{\text{tr}}_{\backslash {\cal B}} (\omega) = %\sigma ,\\
%&\omega,\sigma\geq 0,\, \text{tr}(\omega)=\text{tr}%(\sigma)=1,
%\end{eqnarray}
%
%satisfies $z_{m,s}/m\leq e_\text{min}(H_N)$.
%\end{proposition}
\begin{proof}
\new{This can be easily seen as a relaxation of the
full problem for $s=m$: Then the minimization 
delivers exactly the smallest eigenvalue 
$\lambda_{\text{min}}(h_{2m})$ of $h_{2m}$, so that
\begin{equation}
\frac{x_{m,m}}{(2m-1)} = 
\frac{\lambda_\text{min}(h_{2m})}{(2m-1)}\leq e_\text{min}
\end{equation}
delivers precisely the Anderson bound for $2m$ sites.
For smaller values $s=1,\dots, m-1$, 
one merely requires $\omega$
and $\sigma$ to be identical on a subset of
sites, hence relaxing the problem, leading to smaller and less tight lower bounds,
\begin{equation}
\frac{x_{m,s}}{(2m-1)} \leq \frac{x_{m,t}}{(2m-1)}
\leq \frac{x_{m,m}}{(2m-1)}\leq  e_\text{min}
\end{equation}
for $t=s,\dots, m$.
} 

%This problem is obtained as a convex relaxation of Eqs.~(\ref{fu1})-(\ref{fu3}), by relaxing marginal constraints. 
%As the full problem, it amounts to 
%solving a semi-definite problem \cite{Boyd2004}, but now involving quantum states on $(\cc^d)^{\otimes m}$ only, 
%as can in turn 
%be solved 
%with interior-point methods \cite{Boyd2004,InteriorPoint}.
\end{proof}
\new{Hence, making use of the
quantum marginal problem,
one arrives at a hierarchy of new bounds.}
%Here, the correct marginals are merely enforced, therefore, reminiscent of the  quantum marginal problem \cite{QuantumMarginalProblem,ChristandlMarginal}, instead of the entire quantum states $\sigma$ and $\omega$ being identical. 

%\emph{Outlook.} 

\section{Summary and outlook}
This \new{work} emphasizes that one can easily equip upper bounds to ground states obtained
by resorting to classical or quantum variational principles with concomitant lower bounds: These bounds
certify the quality of the variational ansatz. As technical results, performance guarantees of Anderson bounds
and of semi-definite relaxations are proven. What is more, an improved Anderson bound is presented. 
The upshot is that certified bounds to the energy density up to small constant $O(1)$ in the system size
-- for the Anderson bound even quantitative ones -- are easy to get classically.

%None of the bounds presented is particularly sophisticated, technically speaking. 
\new{All lower bounds are comparably simple.}
One may argue, however, that it is their
simplicity that renders them useful: Again, they can be interpreted as ``de-quantization statements''.
It is sometimes under-appreciated that one can easily classically
approximate ground state energy densities up to a small constant error from below: 
This places stringent demands on quantum simulations aimed at producing such 
approximations. Any quantum simulation aimed at 
approximating ground state energies hence has to deliver approximations that scale more
favourable compared to this in order to possibly outperform classical computations. 
Ref.~\cite{Sev2021} makes a similar point for systems of quantum chemistry, 
 stressing that estimating the ground state energy of a local Hamiltonian when given, as an additional input, 
 a state sufficiently close to the ground state, can be solved efficiently with constant precision on a classical computer. \new{This work also provides a ``proof pocket'',
 providing variational quantum eigensolver with rigorous performance guarantees \cite{MindTheGaps}, and hence contributes to making near-term quantum computing more reliable and quantitative.}
 
The results stated are clearly not in contradiction with the famous
\emph{quantum PCP} conjecture \cite{QuantumPCP}. After all, this conjecture states that it
remains {\tt QMA}-hard to approximate the  
ground state energy even up to an error $\gamma n$ for some absolute constant
$0 < \gamma < 1$, where $n$ is the number of local terms in the Hamiltonian. Obviously, 
 the above bounds produce exactly such an energy approximation, which only once again implies
 that the statement of the quantum PCP conjecture cannot expected to be tight for cubic lattices.
It is also worth noting that  all the mentioned bounds apply equally well to 
\emph{fermionic Hamiltonians} \cite{HastingsFermionic,HelsenFermionic} 
which have again moved to the focus of attention recently,
not the least as the precise scaling of the ground state energy of the 
\emph{Sachdev-Ye-Kitaev} (SYK) model 
\cite{PhysRevLett.70.3339,SYK, SYK2} 
of random degree polynomials has become interesting. It might also be fruitful to compare
the discussed bounds with improvements of Temple's lower bound \cite{Martinazzo}.
It is the hope that this \new{work} can 
contribute to the development of
benchmarks for variational principles in both the classical and quantum reading. 

%{\it Acknowledgements.} 

\subsection*{Acknowledgements}
Discussions with J.~Haferkamp as 
well as comments by D.~Miller and J.~M.~Arrazola
are acknowledged. This work has been funded by the DFG (CRC 183, SPP 2514), 
the \new{BMFTR} 
(MuniQCAtoms, FermiQP, Hybrid, \new{QuSol}, \new{Hybrid++}), the BMWK (EniQmA), the Einstein Foundation, the Munich Quantum Valley (K8), the Quant\-ERA (HQCC), \new{and the European Research Council (DebuQC).}
After completion of this work, the author of this note
became aware of an exciting related but different relaxation of the 
infinite translationally invariant ground state problem
\cite{SchuchRelaxation} compared to what is stated as Proposition 3. 
It would be interesting to compare the two
relaxations, \new{which is left to future work.} 

\bibliographystyle{apsrev4-1}
%\bibliography{BigReferences68}


\end{document}